\newif\ifsiam

\ifsiam
	\documentclass[final]{siamltex}
\newcommand{\qed}{\vbox{\hrule height0.6pt\hbox{%
   \vrule height1.3ex width0.6pt\hskip0.8ex
   \vrule width0.6pt}\hrule height0.6pt}}
\else
	\documentclass[12pt,twoside]{article}
	\usepackage{amsthm}
\newtheorem{theorem}{Theorem}[section]
\newtheorem{lemma}[theorem]{Lemma}
\newtheorem{proposition}[theorem]{Proposition}

	\renewcommand{\and}{ ~and~ }
\usepackage[a4paper, inner=2cm, top=2cm, outer=1.5cm, bottom=1.5cm, includefoot]{geometry}
\fi

\usepackage{amsmath}
\usepackage{amssymb}
\usepackage{bm}
\usepackage{arydshln}
\usepackage{graphicx}
\usepackage{subfigure}
\usepackage{yfonts}

\numberwithin{equation}{section}
    
\usepackage{color}
\definecolor{alertColor}{rgb}{0.7, 0, 0}

\newcommand{\br}{{r}}

\newcommand{\balpha}{{\alpha}}
\newcommand{\bbeta}{{\beta}}
\newcommand{\calA}{{\mathcal A}}
\newcommand{\calB}{{\mathcal B}}
\newcommand{\calE}{{\mathcal E}}
\newcommand{\calI}{{\mathcal I}}

\newcommand{\calM}{{\mathcal M}}
\newcommand{\calN}{{\mathcal N}}
\newcommand{\calS}{{\mathcal S}}
\newcommand{\calU}{{\mathcal U}}
\newcommand{\calV}{{\mathcal V}}
\newcommand{\bbC}{{\mathbb C}}
\newcommand{\bbN}{{\mathbb N}}
\newcommand{\bbP}{{\mathbb P}}
\newcommand{\bbQ}{{\mathbb Q}}
\newcommand{\bbR}{{\mathbb R}}

\renewcommand{\Re}{{\rm Re}}
\renewcommand{\Im}{{\rm Im}}
\renewcommand{\i}{{\rm i}}

\newcommand{\perm}{{\rm perm}}
\newcommand{\rot}{{\rm rot}}
\newcommand{\aprod}[1]{\sideset{^#1}{}\prod}
\newcommand{\tot}{{\rm tot}}
\newcommand{\q}{{\rm q}}
\newcommand{\tb}{{\rm q}}
\newcommand{\cut}{{\rm cut}}
\newcommand{\pair}{{\rm pair}}
\newcommand{\eam}{{\rm eam}}
\newcommand{\sym}{{\rm sym}}

\newcommand{\tildeM}{M}
\newcommand{\Ncap}{N_{\rm cap}}

\begin{document}

\title{Moment Tensor Potentials: a class of systematically improvable interatomic potentials}
\author{Alexander V. Shapeev%
\thanks{Skolkovo Institute of Science and Technology, Skolkovo Innovation Center, Building 3, Moscow 143026 Russia. Email: {\tt a.shapeev@skoltech.ru}\,.}
}
\maketitle

\begin{abstract}
Density functional theory offers a very accurate way of computing materials properties from first principles.
However, it is too expensive for modelling large-scale molecular systems whose properties are, in contrast, computed using interatomic potentials.

The present paper considers, from a mathematical point of view, the problem of constructing interatomic potentials that approximate a given quantum-mechanical interaction model. In particular, a new class of systematically improvable potentials is proposed, analyzed, and tested on an existing quantum-mechanical database.
\end{abstract}


\section{Introduction}

Molecular modelling is an increasingly popular tool in biology, chemistry, physics, and materials science \cite{Griebel-book}.
The success of molecular modelling largely depends on the accuracy and efficiency of calculating the interatomic forces.
The two major approaches to computing interatomic forces are (1) quantum mechanics (QM) calculations \cite{Cances2003book,Finnis2003book}, and (2) using empirical interatomic potentials.
In the first approach, first the electronic structure is computed and then the forces on the atoms (more precisely, on their nuclei) are deduced.
QM calculations are therefore rather computationally demanding, but can yield a high quantitative accuracy.
On the other hand, the accuracy and transferability of empirical interatomic potentials is limited, but they are orders of magnitude less computationally demanding.
With interatomic potentials, typically, the forces on atoms derive from the energy of interaction of each atom with its atomic neighborhood (typically consisting of tens to hundreds of atoms).

This hence makes it very attractive to design a combined approach whose efficiency would be comparable to interatomic potentials, yet the accuracy is similar to the one obtained with ab initio simulations.
Moreover, the scaling of computational complexity of using the interatomic potentials is $O(N)$, where $N$ is the number of atoms, whereas, for instance, standard Kohn-Sham density functional theory calculations scale like $O(N^3)$---which implies that the need in such combined approaches is even higher for large numbers of atoms.

One way of achieving this is designing accurate interatomic potentials.
With this goal in mind, we categorize all potentials into two groups, following the regression analysis terminology of \cite{BartokKondorCsanyi2013descriptors}.
The first group is the parametric potentials with a fixed number of numerical or functional parameters.
All empirical potentials known to the author are parametric, which is their disadvantage---their accuracy cannot be systematically improved.
As a result, if the class of problems is sufficiently large and the required accuracy is sufficiently high (say, close to that of ab initio calculations), then the standard parametric potentials are not sufficient and one needs to employ nonparametric potentials.
In theory, the accuracy of these potentials can be systematically improved.
In practice, however, increasing the accuracy is done at a cost of higher computational complexity (which is nevertheless orders of magnitude less than that of the QM calculations), and the accuracy is still limited by that of the available QM model the potential is fitted to.
Also, it should be noted that, typically, such fitted potentials do not perform well for the situations that they were not fitted to (i.e., they exhibit little or no transferability).

Each nonparametric potential has two major components: (1) a representation (also referred to as ``descriptors'') of atomic environments, and (2) a regression model which is a function of the representation.
It should be emphasized that finding a good representation of atomic environments is not a trivial task, as the representation is required to satisfy certain restrictions, namely, invariance with respect to Euclidean transformations (translations, rotations, and reflections) and permutation of chemically equivalent atoms, smoothness with respect to distant atoms coming to and leaving the boundary of the atomic environment, as well as completeness (for more details see the discussion in \cite{BartokKondorCsanyi2013descriptors,Behler2011symmetry_functions,Ferre2015permutation}).
The latter restriction means that the representation should contain all the information to unambiguously reconstruct the atomic environment.

In this context, one common approach to constructing nonparametric potentials, called the neural networks potentials (NNP) \cite{Behler2014review,BehlerParrinello2007NN}, is using artificial neural networks as the regression model and a family of descriptors first proposed by Behler and Parrinello \cite{BehlerParrinello2007NN}.
These descriptors have a form of radial functions applied to and summed over atomic distances to the central atom of the environment, augmented with similar functions with angular dependence summed over all pairs of atoms.
Another approach called Gaussian Approximation Potentials (GAP) \cite{2010GAP,2014GAP_tungsten} uses Gaussian process regression and a cleverly designed set of descriptors that are evaluated by expanding smoothened density field of atoms into a spherical harmonics basis \cite{BartokKondorCsanyi2013descriptors}.

Other recent approaches include \cite{ThompsonEtAl2014snap} which employs the bispectrum components of the atomic density, as proposed in an earlier version of GAP \cite{2010GAP}, and uses a linear regression model to fit the QM data.
Also, \cite{LiKermodeDevita2015learning_forces} uses Gaussian process regression of a force-based model rather than the potential energy-based model.
Additionally, in a recent work \cite{Mallat2015} the authors put forward an approach of regressing the interatomic interaction energy together with the electron density.

It is worthwhile to mention a related field of research, namely the development of non-reactive interatomic potentials.
Such potentials assume a fixed underlying atomistic structure (fixed interatomic bonds or fixed lattice sites that atoms are bound to).
The works developing nonparametric versions of such potentials include \cite{Ai2014slave,Hellman2013TDEP,Nelson2013cluster}.

In the present paper we propose a new approach to constructing nonparametric potentials based on linear regression and invariant polynomials.
The main feature of this approach is that the proposed form of the potential can provably approximate any regular function satisfying all the needed symmetries (see Theorems \ref{th:main} and \ref{th:conv}), while being computationally efficient.
The building block of the proposed potentials is what we call the moment tensors---similar to inertia tensors of atomic environments.
Hence we call this approach the Moment Tensor Potentials (MTP).

The manuscript is structured as follows.
Section \ref{sec:IP} gives an overview of interatomic potentials.
Section \ref{sec:MTP} introduces MTP and then formulates and proves the two main theorems.
The practical implementation of MTP is discussed in Section \ref{sec:implementation}.
Section \ref{sec:performance-tests} reports the accuracy and computational efficiency (i.e., the CPU time) of MTP and compares MTP with GAP.

\section{Interatomic Potentials}\label{sec:IP}

Consider a system consisting of $N_\tot$ atoms with positions $x\in(\bbR^d)^{N_\tot}$ (more precisely, with positions of their nuclei to be $x$), where $d$ is the number of physical dimensions, taken as $d=3$ in most of applications.
These atoms have a certain energy of interaction, $E^\q(x)$, typically given by some QM model that we aim to approximate.
For simplicity, we assume that all atoms are chemically equivalent.

We make the assumption that $E^\q(x)$ can be well approximated by a sum of energies of atomic environments of the individual atoms.
We denote the atomic environment of atom $k$ by $Dx_k$ and let it equal to a tuple
\begin{equation}\label{eq:D_def}
Dx_k := (x_i-x_k)_{1\leq i\leq N_\tot ,\, 0<|x_i-x_k|\leq R_\cut},
\end{equation}
where $R_\cut$ is the cut-off radius, in practical calculations typically taken to be between $5$ and $10$\AA.
In other words, $Dx_k$ is the collection of vectors joining atom $k$ with all other atoms located at distance $R_\cut$ or closer.
Thus, we are looking for an approximant of the form
\begin{equation}\label{eq:E_def}
E(x) := \sum_{k=1}^{N_\tot} V(Dx_k),
\end{equation}
where $V$ is called the interatomic potential.
This assumption is true in most systems with short-range interactions (as opposed to, e.g., Coulomb interaction in charged or polarized systems), refer to recent works \cite{ChenOrtner1,ChenOrtner2,NazarOrtner} for rigorous proofs of this statement for simple QM models.

Mathematically, since $Dx_k$ can be a tuple of any size (in practice limited by the maximal density of atoms), $V$ can be understood as a family of functions each having a different number of arguments.
For convenience, however, we will still refer to $V$ as a ``function''.

The function $V=V(u_1,\ldots,u_n)$ is required to satisfy the following restrictions:
\begin{description}
\item[(R1)] Permutation invariance:
\[
V(u_1,\ldots,u_n) = V(u_{\sigma_1},\ldots,u_{\sigma_n})
\qquad\text{for any $\sigma\in \calS_n$},
\]
where $\calS_n$ denotes the set of permutations of $(1,\ldots,n)$.

\item[(R2)] Rotation and reflection invariance:
\[
V(u_1,\ldots,u_n) = V(Q u_1,\ldots,Q u_n)
\qquad\text{for any $Q\in O(d)$},
\]
where $O(d)$ is the orthogonal group in $\bbR^d$.
For simplicity in what follows we will denote $Qu := (Q u_1,\ldots,Q u_n)$.

\item[(R3)] Smoothness with respect to the number of atoms (more precisely, with respect to atoms leaving and entering the interaction neighborhood):
\begin{align*}
V & : \bbR^{d\times n}\to \bbR\text{ is a smooth function}
\\
V(u_1,\ldots,u_n)
&=
V(u_1,\ldots,u_n,u_{n+1})
\qquad\text{whenever $|u_{n+1}|\geq R_\cut$}.
\end{align*}
In most of the works performing practical calculations, including \cite{Behler2011symmetry_functions,Behler2014review,BehlerParrinello2007NN,Szlachta2014thesis,2014GAP_tungsten,ThompsonEtAl2014snap}, the interatomic potentials are chosen such that the energy and forces are continuous. There are, however, exceptions including \cite{2010GAP} that require the second derivatives of energy to be continuous, and \cite{jalkanen2015systematic} proposing exponential decay of the potential with no strict cut-off radius.
\end{description}
Note that $E$ is translation symmetric by definition.

\subsection{Empirical Interatomic Potentials}

For the purpose of illustration, we will briefly present two popular classes of empirical interatomic potentials. The first class is pair potentials (also known as two-body potentials),
\[
V^\pair(u) = \sum_{i} \varphi(|u_i|),
\]
where by $u$ we denote the collection of the relative coordinates, $u=(u_i)_{i=1}^{n}$.
It is clear that this potential satisfies {\bf (R1)} and {\bf (R2)}, while if we take $\varphi(r)=0$ for $r\geq R_\cut$ (as is done in most of practical calculations) then it also satisfies {\bf (R3)}.
The second potential is the Embedded Atom Model (EAM),
\[
V^\eam(u) = \sum_{i} \varphi(|u_i|)
+
F\Big({\textstyle
	\sum_{i} \rho(|u_i|)
}\Big),
\]
where if $\varphi(r)$ and $\rho(r)$ are chosen to vanish for $r \geq R_\cut$ then it also satisfies {\bf (R1)}--{\bf (R3)}.

Both potentials can be viewed to have descriptors of the form
\begin{equation}\label{eq:desc_radial}
R_\nu(u) = \sum_{i} f_\nu(|u_i|),
\end{equation}
where $f_\nu$ are chosen such that their linear combination can approximate any smooth function (e.g., $f_\nu(r) = r^\nu (R_\cut-r)^{2}$
for $r\leq R_\cut$ and $f_\nu(r)=0$ for $r>R_\cut$, where $\nu=0,1,\ldots$; the term $(R_\cut-r)^{2}$ ensures continuous energy and forces).
Then, we can approximate $\varphi(r) \approx \sum_{\nu} c_\nu f_\nu(r)$ (the sum is taken over some finite set of $\nu$), and hence
\[
V^\pair(u) \approx \sum_{\nu} c_\nu R_\nu(u).
\]
Likewise one can approximate $V^\eam$ with the exception that it will not be, in general, a linear function of $R_\nu(u)$.

This set of descriptors is not complete, because it is based only on the distances to the central atom, but is insensitive to bond angles.

\subsection{Non-empirical Interatomic Potentials}

Next, we review a number of non-empirical interatomic potentials proposed recently as a more accurate alternative to the empirical ones.

We start with the NNPs \cite{Behler2011symmetry_functions,Behler2014review,BehlerParrinello2007NN}.
In addition to the descriptors given by \eqref{eq:desc_radial} they also employ the descriptors of the form
\begin{equation}\label{eq:desc_three_body}
\sum_{i=1}^n \sum_{\substack{j=1 \\ j\ne i}}^n f(|u_i|,|u_j|, u_i \cdot u_j).
\end{equation}
(Note that NN can be, in principle, used with any set of descriptors. If these descriptors are complete then such NNP is systematically improvable.)
Typically, 50 to 100 descriptors of the form \eqref{eq:desc_radial} and \eqref{eq:desc_three_body} are chosen as an input to a NN, while its output yields the function $V$ \cite{Behler2014review}.
In practice, this approach gives convincing results, however, it is an open question whether or not these descriptors are complete.

The GAP \cite{2010GAP,2014GAP_tungsten} uses a different idea, consisting of: (1) forming a smoothened atomic density function
\[
\sum_{i=1}^n \exp\big({\textstyle -\frac{|u_i-u|^2}{2\sigma^2}}\big),
\]
(2) approximating it through spherical harmonics, and (3) constructing functionals applied to the spherical harmonics coefficients that satisfy all needed symmetries.
One can prove mathematically that this approach can approximate any regular symmetric function (i.e., a function that satisfies {\bf (R1)}--{\bf (R3)}).
However, expanding functions in a spherical harmonics basis can be computationally expensive.

Several other non-empirical potentials have been proposed recently.
A method using a representation based on expanding the atomic density function in spherical harmonics together with linear regression was used in \cite{ThompsonEtAl2014snap}.
In \cite{LiKermodeDevita2015learning_forces} the authors use Gaussian process regression to train a model that predicts forces on atoms directly (as opposed to a model that fits the energy).

\section{Moment Tensor Potentials}\label{sec:MTP}

\subsection{Representation with Invariant Polynomials}\label{sec:representation}

In the present paper we propose an alternative approach based on invariant polynomials.
The idea is that any given potential $V^*(u)$, if it is smooth enough, can be approximated by a polynomial $p(u)\approx V^*(u)$---we will analyze the error of such an approximation in Section \ref{sec:error}.
The approximant can always be chosen symmetric.
Indeed, given $p(u)$, one can consider the symmetrized polynomial
\begin{equation}\label{eq:p_symm}
p^\sym(u_1,\ldots,u_n) = \frac{1}{n!} \sum_{\sigma\in\calS_n} p(u_{\sigma_1},\ldots,u_{\sigma_n}),
\end{equation}
and then $V^*(u)\approx p^\sym(u)$.
Similarly, one can symmetrize $p^\sym$ with respect to rotations and reflections.
Hence, theoretically, one can construct a basis of such polynomials $b_\nu(u)$ and choose
\[
V^*(u) \approx V(u) := \sum_\nu c_\nu b_\nu(u).
\]
This approach is implemented for small systems of up to ten atoms \cite{BraamsBowman2009inv_poly}, however, generalizations of this approach require a more efficient way of generating the invariant polynomials.
The main difficulty is related to the fact that the number of permutations in a system of $n$ atoms is $n!$ which grows too fast in order to, for instance, calculate the right hand side of \eqref{eq:p_symm}.

In the present paper we propose a basis for the set of all polynomials invariant with respect to permutation, rotation, and reflection.
The main feature of the proposed basis is that the computational complexity of computing these polynomials scales like $O(n)$.
Moreover, one can easily construct such bases to also satisfy the {\bf (R3)} property (refer to Section \ref{sec:implementation}), making it a promising candidate for efficient nonparametric interatomic potentials.

\subsubsection*{The $M$ polynomials}

The building blocks of the basis functions for the approximation (representation) of $V=V(u)$ are the ``moment'' polynomials $M = M_{\bullet,\bullet}(u)$ defined as follows.

For integer $\mu,\nu\geq 0$ we let
\[
M_{\mu,\nu}(u) := \sum_{i=1}^{n} |u_i|^{2\mu} u_i^{\otimes \nu},
\]
where $w^{\otimes \nu} := w\otimes\ldots\otimes w$ is the Kronecker product of $\nu$ copies of the vector $w\in\bbR^d$.
Thus, $M_{\mu,\nu}(u) \in (\bbR^d)^\nu$ (i.e., an $\nu$-dimensional tensor) for each $u$.
Computing $M_{\mu,\nu}(u)$ requires linear time in $n$, but exponential in $\nu$.
This means that if the maximal value of $\nu$ is bounded
then computing $M_{\mu,\nu}$ can be done efficiently.

There is a mechanical interpretation of $M_{\mu,\nu}(u)$.
Consider $\mu=0$, then $M_{0,0}$ simply gives the number of atoms at the distance of $R_\cut$ or less (this can also be understood as the ``mass'' of these atoms), $M_{0,1}$ is the center of mass of such atoms (scaled by the mass), $M_{0,2}$ is the tensor of second moments of inertia, etc.
For $\mu>0$, $M_{\mu,\nu}$ can be interpreted as weighted moments of inertia, with $i$-th atom's weight being $|u_i|^{2\mu}$.

\subsubsection*{The basis polynomials $B_\balpha$}

The basis polynomials are indexed by $k\in\bbN$, $k\leq n$, where our definition for $\bbN$ is
\[
\bbN=\{0,1,\ldots\},
\]
and a $k\times k$ symmetric matrix $\alpha$ of integers $\alpha_{i,j}\geq 0$ ($i,j\in\{1,\ldots,k\}$).
For such matrices, by $\alpha_i'$ we define the sum of the off-diagonal elements of $i$-th row,
\begin{equation}
\label{eq:alpha-prime}
\alpha_i' = \sum_{\substack{j=1 \\ j\ne i}}^{k} \alpha_{i,j}.
\end{equation}
Next, we define a contraction operator (product) of tensors $T^{(i)}\in(\bbR^d)^{\alpha_{i}'}$ by
\begin{align*}
\aprod\balpha_{i=1}^k T^{(i)}
&:=
\sum_{\beta}
\prod_{i=1}^k
T^{(i)}_{\beta^{(1,i)}\ldots\beta^{(i-1,i)}\beta^{(i,i+1)}\ldots\beta^{(i,k)}},
\end{align*}
where each $\beta$ is a collection of multiindices
$\beta=\big(\beta^{(i,j)}\big)_{1\leq i < j \leq k}$,
and each multiindex $\beta^{(i,j)}$ has the following form,
\[
\beta^{(i,j)} = \big(\beta^{(i,j)}_1,\ldots,\beta^{(i,j)}_{\alpha_{i,j}}\big)\in\{1,\ldots,d\}^{\alpha_{i,j}} \quad 1\leq i < j \leq k.
\]
To define this contraction rigorously, we let
\begin{equation}\label{eq:calMcalB}
\calM := \{(i,j)\in \{1,\ldots,m\}^2 : i<j\},
\qquad\text{and}\qquad
\calB := (\{1,\ldots,d\}^{\alpha_{ij}})_{(i,j)\in\calM}
\end{equation}
and, 
expanding the multiindex notation, we write
\begin{align*}
&
\aprod\balpha_{i=1}^k T^{(i)}
=
\sum_{\beta\in\calB}
\prod_{i=1}^k
T^{(i)}_{
	\beta^{(1,i)}_1\ldots\beta^{(1,i)}_{\alpha_{1i}}
	~~\ldots~~
	\beta^{(i-1,i)}_1\ldots\beta^{(i-1,i)}_{\alpha_{i-1,i}}
	~~
	\beta^{(i,i+1)}_1\ldots\beta^{(i,i+1)}_{\alpha_{i,i+1}}
	~~\ldots~~
	\beta^{(i,k)}_1\ldots\beta^{(i,k)}_{\alpha_{i,k}}}
.
\end{align*}
Hence $\alpha_{ij}$ can be interpreted as how many dimensions are contracted between $T^{(i)}$ and $T^{(j)}$.

Finally, we let
\begin{equation}\label{eq:Balpha}
B_{\balpha}(u)
:=
\aprod\balpha_{i=1}^k M_{\alpha_{ii},\alpha_i'}(u)
,
\end{equation}
and call it a basis function (for a given $\balpha$).
Here $B_{\balpha}(u)$ is, essentially, a $k$-body function.
We note that $B_{\balpha}(u) \equiv B_{\bbeta}(u)$ if there exists a permutation $\sigma\in\calS_k$ such that $\alpha_{ij} = \beta_{\sigma_i\sigma_j}$ for all $i$ and $j$.

For illustrative purposes
we work out three examples of different $B_\alpha$.
First, let us take
$\alpha=\begin{pmatrix}
\mu_1 & 1 \\ 1 & \mu_2
\end{pmatrix}.$
This implies $k=2$ (since this is a $2\times 2$ matrix) and $\alpha_1' = \alpha_2' = 1$. Hence $B_\alpha$ is the scalar product of $M_{\mu_1,1}$ and $M_{\mu_2,1}$, indeed:
\[
B_\alpha =
\sum_{\beta^{(1,2)}_1 = 1}^{d}
	\big(M_{\mu_1,1}\big)_{\beta^{(1,2)}_1}
	\big(M_{\mu_2,1}\big)_{\beta^{(1,2)}_1}
= M_{\mu_1,1} \cdot M_{\mu_2,1},
\]
where $\big(M_{\mu_1,1}\big)_{\beta^{(1,2)}_1}$ denotes the $\beta^{(1,2)}_1$-th component of the vector $M_{\mu_1,1}$.

As the next example, we take
$\alpha=\begin{pmatrix}
\mu_1 & 2 \\ 2 & \mu_2
\end{pmatrix}$.
Then $\alpha_1'=\alpha_2'=2$, and hence $B_\alpha$ is the Frobenius product of the two matrices, $M_{\mu_1,2}$ and $M_{\mu_2,2}$:
\[
B_\alpha =
\sum_{\beta^{(1,2)}_1 = 1}^{d}
\sum_{\beta^{(1,2)}_2 = 1}^{d}
	\big(M_{\mu_1,2}\big)_{\beta^{(1,2)}_1,\beta^{(1,2)}_2}
	\big(M_{\mu_2,2}\big)_{\beta^{(1,2)}_1,\beta^{(1,2)}_2}
= M_{\mu_1,1} \!:\! M_{\mu_2,1}.
\]
In the last example, we take
$\alpha=\begin{pmatrix}
\mu_1 & 1 & 1 \\ 1 & \mu_2 & 0 \\ 1 & 0 & \mu_3
\end{pmatrix}$.
Then $B_\alpha$ is the following contraction of the matrix $M_{\mu_1,2}$, and two vectors, $M_{\mu_2,1}$ and $M_{\mu_2,1}$:
\[
B_\alpha =
\sum_{\beta^{(1,2)}_1 = 1}^{d}
\sum_{\beta^{(1,3)}_1 = 1}^{d}
	\big(M_{\mu_1,2}\big)_{\beta^{(1,2)}_1,\beta^{(1,3)}_1}
	\big(M_{\mu_2,1}\big)_{\beta^{(1,2)}_1}
	\big(M_{\mu_3,1}\big)_{\beta^{(1,3)}_1}
= \big(M_{\mu_1,2} M_{\mu_2,1}\big)\cdot M_{\mu_3,1}.
\]

\subsubsection*{Representability}

We are now ready to formulate the main result, namely that the linear combinations of $B_\alpha$ span all permutation and rotation invariant polynomials.
Let us denote the set of all polynomials of $n$ vector-valued variables by $\bbP$, the set of permutation invariant polynomials by $\bbP_{\perm}\subset\bbP$ and the set of rotation invariant polynomials by $\bbP_{\rot}\subset\bbP$.

\begin{theorem}\label{th:main}
The polynomials $B_\balpha$ form a spanning set of the linear space $\bbP_\rot\cap\bbP_\perm \subset \bbP$, in the sense that any $p\in\bbP_\rot\cap\bbP_\perm$ can be represented by a (finite) linear combination of $B_\balpha$ (but this combination is, in general, not unique).
\end{theorem}

We postpone the proof of this result to Section \ref{sec:proof}, after we illustrate in Section \ref{sec:error} that a QM model can be efficiently approximated with polynomials.


\subsection{Approximation Error Estimate}\label{sec:error}

In this section we present an error analysis of fitting the prototypical tight-binding QM model as proposed in \cite{ChenOrtner1,ChenOrtner2} with the polynomials $B_\alpha(u)$.
We note that we directly take the site energy $V^\tb(u)$ constructed in \cite{ChenOrtner1,ChenOrtner2} (rather than starting from the total energy $E^\tb(x)$) and fix $n$ to be constant throughout this analysis.
The latter assumption needs some comment.
Letting $n$ to be constant is essentially equivalent to assuming finite $R_\cut$, provided that the atoms cannot come too close to each other.
We note that making such assumption is not an insignificant limitation of the present analysis, however, we find lifting this limitation not at all trivial.

The QM model is defined as follows. We let, for convenience, $u_0 := 0$ and define the Hamiltonian matrix
\[
H_{ij}(u) := \begin{cases}
\varphi(u_j-u_i) & \quad i,j \in \{0,\ldots,n\} ~~\text{and}~~ i\ne j, \\
0 & \quad i=j\in\{0,\ldots,n\},
\end{cases}
\]
where $\varphi:\bbR^d\to\bbR$ is an empirically chosen function referred to as the hopping integral \cite{Finnis2003book}.

Then we apply the function $f^\tb$ to the matrix $H$ (refer to \cite{Higham2008functions} for details on matrix functions) and define $V^\tb$ as the $(0,0)$-th element of this matrix:
\[
V^\tb(u) := (f^\tb(H))_{0,0},
\]
where
\[
f^\tb(\epsilon) := \epsilon\,\Big(1+e^{\epsilon/(k_{\rm B} T)}\Big)^{-1},
\]
where $k_{\rm B}>0$ is the Boltzmann constant, and $T>0$ is the electronic temperature, and we take the chemical potential to be zero.

Next, we let $R>0$ and $\delta_0>0$, allow $u_i$ to vary in $\calV_{\delta_0}$, where \[
\calV_{\delta} := \{\zeta\in\bbC : \Re(\zeta)\in[-R-\delta,R+\delta]^d ,\, \Im(\zeta)\in[-\delta,\delta]^d \},
\]
and assume that $\varphi(v)$ is analytically extended to $\calV_{\delta_0}$.
We note that in many models $\varphi(\zeta)$ has a singularity at $\zeta=0$ (for example, $\varphi(\zeta) = \beta_0 \exp(-q |\zeta|)$ for some $\beta_0$, $q>0$ \cite[Equation (7.24)]{Finnis2003book}), therefore by assuming analytical extensibility of $\varphi$ onto $\calV_{\delta_0}$ we implicitly assume some approximation of such irregular $\varphi$ with a function that is regular around $\zeta=0$.
For example, one could use the Hermite functions basis, $e^{-|\zeta|^2/2} H_{n_1}(\zeta_1) H_{n_2}(\zeta_2) \ldots H_{n_d}(\zeta_d)$ \cite{Szego1939}, to approximate $\varphi(\zeta)$ away from $\zeta=0$.


\begin{theorem}\label{th:conv}
There exist $C>0$ and $\rho>1$, both depending only on $n$, $\delta_0$, $M_{\delta_0}$, and $k_{\rm B} T$, such that for any $m\in\bbN$ there exists $p_m\in\bbP_\perm\cap\bbP_\rot$ of degree $m$ such that
\begin{equation}\label{eq:conv}
\sup_{u:\,\max_i |u_i|\leq R} |V^\tb(u) - p_m(u)| < C \rho^{-m}.
\end{equation}
\end{theorem}
The proof of this theorem is based on the following result by Hackbusch and Khoromskij \cite{HackbuschKhoromskij2001}:
\begin{proposition}\label{prop:Khoromskij}
Let $\calE_\rho := \{\zeta\in\bbC : |z-1|+|z+1| \leq \rho+\rho^{-1}\}$, for some $\rho>1$, 
\[
\calE^{(j)}_\rho := [-1,1]^{j-1} \times \calE_\rho \times [-1,1]^{n-j-1},
\]
$f=f(z_1,\ldots,z_n)$ defined on the union of all $\calE^{(j)}_\rho$, and
\[
M_\rho(f) := \max_{1\leq j\leq n} \sup_{z\in\calE^{(j)}_\rho} |f(z)|.
\]
Let $p=p(z_1,\ldots,z_n)$ be the polynomial interpolant of $f$ of degree $m$ on the Chebyshev-Gauss-Lobatto nodes in $[-1,1]^n$.
Then
\[
\sup_{\max_i |z_i|\,\leq 1} |f(z) - p(z)| \leq 2 (1+2\pi^{-1} \log(n))^m M_\rho(f) \frac{\rho^{-m}}{\rho-1}.
\]
\qed
\end{proposition}

An illustration of $\calE_\rho$ and $\calV_{\delta}$ is given in Figure \ref{fig:complex_regions}.
\begin{figure}[htb]
\hfill
\includegraphics[scale=1]{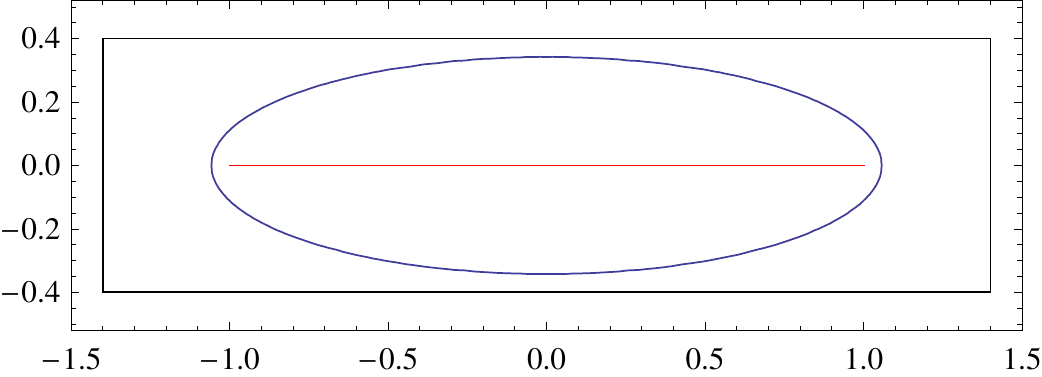}
\hfill
$\mathstrut$
\caption{%
Illustration of $\calV_\delta$ (black square) and $\calE_\rho$ (blue ellipse) on the complex plane for $R=1$, $\delta=0.4$, $\rho=1.4$. Both regions contain [-1,1] (red line) and for this choice of parameters $\calE_\rho \subset \calV_\delta$.
}
\label{fig:complex_regions}
\end{figure}

\begin{proof}[Proof of Theorem \ref{th:conv}]
The plan of the proof is to, after an auxiliary step (Step 1), successively obtain bounds on $\varphi$ (Step 2), on $H(u)$ (Step 3), and on the interpolating polynomial (Step 4), and then symmetrize this interpolating polynomial (Step 5).

\medskip\noindent
{\it Step 1.} First, we define
$M_{\delta} := \sup_{\zeta\in \calV_{\delta}} \varphi(\zeta)$
and note that by the Cauchy integral formula one can bound
\[
|\varphi'(z)| \leq
\left|
\oint_{\partial \calV_{2\delta}} \frac{\varphi(\zeta) {\rm d}\zeta}{(\zeta-z)^2}
\right|
=
\frac{(2 R+8\delta) M_{2\delta}}{\delta^2} =: M'_{\delta}
\qquad \forall z:\,|z|\leq\delta
.
\]
This is valid for any $\delta\leq\delta_0/2$.

\medskip\noindent
{\it Step 2.}
Next we let $\delta \in (0,\delta_0/2]$, which will be fixed later, and note that if $z\in\calV_{\delta}$ then
$|\Im(\varphi(z))| \leq \Im(\varphi(\Re(z))) + M'_{\delta_0} \Im(z) \leq M'_{\delta_0} \delta$ thanks to the intermediate value theorem.

\medskip\noindent
{\it Step 3.}
Next, following Proposition \ref{prop:Khoromskij}, define
\[
\calU_{\delta} := \cup_{i=1}^{n}
\calV_{0}^{i-1} \times \calV_{\delta} \times
\calV_{0}^{n-i-1}.
\]
Hence note that for $u\in\calU_{\delta}$, $H(u)$ is symmetric (possibly non-Hermitian) with at most $2n$ elements being non-real, which makes it easy to estimate using the Frobenius norm:
\[
\|\Im (H(u))\| \leq \sqrt{2n} M'_{\delta} \delta.
\]
By a spectrum perturbation argument (namely, the Bauer-Fike theorem \cite{BauerFike}) we have the corresponding bound on the spectrum:
\[
|\Im({\rm Sp}(H(u)))| \leq \sqrt{2n} M'_{\delta} \delta.
\]
The real part of the spectrum can be estimated directly through the norm:
\[
|\Re({\rm Sp}(H(u)))| \leq \|H(u)\| \leq n M_{\delta_0}.
\]

\medskip\noindent
{\it Step 4.}
We next bound $V^\tb(u)$ for $u\in\calU_{\delta}$.
If needed we decrease $\delta$ such that $\sqrt{2n} M'_{\delta} \delta < \frac{\pi}{3} k_{\rm B} T$ and use the following representation \cite{ChenOrtner1, Higham2008functions}:
\[
V^\tb(u) = -\frac{1}{2\pi\i} \oint_\gamma f^\tb(z) \big((H(u)-zI)^{-1}\big)_{0,0} {\rm d}z,
\]
where we take $\gamma := \partial \Omega$, where $\Omega=\{z\in\bbC : |\Re(z)|\leq nM_{\delta_0}+\frac{\pi}{3} k_{\rm B} T ,\, |\Im(z)| \leq \frac{2\pi}{3} k_{\rm B} T \}$.
The choice of the region $\Omega$ is such that any point $z\in\gamma$ is separated from any eigenvalue $\lambda$ of $H(u)$ by the distance $|z-\lambda|\geq \frac{\pi}{3} k_{\rm B} T$, and at the same time $f^\tb(z)$ is regular on $\overline{\Omega}$.
This allows us to estimate,
for $z\in\gamma$,
\[
\|(H(u)-zI)^{-1}\| \leq \big({\textstyle \frac{\pi}{3}} k_{\rm B} T\big)^{-1}
\]
and hence
\[
\sup_{u\in\calU_\delta} |V^\tb(u)| \leq \sup_{z\in\Omega} |f^{\tb}(z)| \, |\gamma| \, \big({\textstyle \frac{\pi}{3}} k_{\rm B} T\big)^{-1}
\leq C,
\]
where $C$ is a generic constant that depends only on $n$, $\delta_0$, $M_{\delta_0}$, and $k_{\rm B} T$.
It now remains to notice that there exists $\rho>1$ that depends on $\delta$ such that
\[
\{z\in\bbC : |z/R-1|+|z/R+1| \leq \rho+\rho^{-1}\}
\subset
\calV_{\delta}
.
\]
Hence Proposition \ref{prop:Khoromskij} applies to the function $f(z) = V^{\tb}(Rz)$ and yields the interpolating polynomial $\tilde{p}_m(u)$ such that
\[
\sup_{u:\,\max_i |u_i|\leq R} |V^\tb(u) - \tilde{p}_m(u)|
\leq
\sup_{u:\,\max_i |u_i|_\infty\leq R} |V^\tb(u) - \tilde{p}_m(u)|
\leq C \rho^{-m}.
\]
By construction $\tilde{p}_m\in \bbP_\perm$.
Indeed, the function $f(z) = V^{\tb}(Rz)$ is symmetric with respect to permutation of variables, and so is the Chebyshev-Gauss-Lobatto interpolations nodes on the domain $[-1,1]^n$. Hence, uniqueness of interpolation yields permutation symmetry of $\tilde{p}_m$.

\noindent {\it Step 5.}
We define
\[
p_m(u) := \frac{\int_{Q\in O(d)} \tilde{p}_m(Q u) {\rm d}Q}{\int_{Q\in O(d)} {\rm d}Q}
\in \bbP_\perm\cap\bbP_\rot
,
\]
(where ${\rm d}Q$ denotes the Haar measure)
and thanks to the rotation invariance of $V^\tb$ we recover \eqref{eq:conv}.
\end{proof}

It is worthwhile noting that the integration with respect to rotation was used only as a technical tool in the proof.
If we directly constructed an approximation $V(u)=\sum_{\alpha} c_\alpha B_\alpha(u)$ to $V^\tb(u)$, rather than using the Chebyshev-Gauss-Lobatto interpolation as an intermediate step, then $V(u)$ would be rotationally invariant by construction.

\subsection{Proof of Theorem \ref{th:main}}\label{sec:proof}


We start with stating the First Fundamental Theorem for the orthogonal group ${\rm O}(d)$ \cite{Weyl}:
\begin{theorem}
$p\in\bbP$ is rotation invariant if and only if it can be represented as a polynomial of $n (n+1)/2$ scalar variables of the form $r_{ij}(u) := u_i\cdot u_j$, where $1\leq i\leq j\leq n$.
\end{theorem}

We hence can identify a polynomial $p=p(u)\in\bbP_\rot$ with the respective polynomial $q=q(\br)\in\bbQ$, where $\bbQ$ is the set of polynomials of $n (n+1)$ scalar variables $\br=(r_{ij})_{1\leq \,i,j\,\leq n}$ (we lift the requirement $i\leq j$ for the ease of notation in what follows).


In order to proceed, we introduce the notation of composition of tuples:
\begin{equation}\label{eq:tuples_composition}
a_b := (a_{b_1}, \ldots, a_{b_m})
\qquad \text{for}~~ b=(b_1,\ldots,b_m) \subset \{1,\ldots,n\}^m \text{~~and~~} a=(a_1,\ldots,a_n).
\end{equation}
Hence define
\[
\bbQ_\perm := \{q\in\bbQ : q(r) \equiv q(r_{\sigma\sigma}) ~\forall \sigma\in\calS_n\}
\]
that corresponds to $\bbP_\perm$, where we likewise let $r_{\sigma\sigma} := (r_{\sigma_i\sigma_j})_{1\leq \,i,j\,\leq n}$.

%

We next formulate a rather intuitive result, essentially stating that $\bbQ_\perm$ can be spanned by symmetrizing all monomials of $\br$.
\begin{lemma}
\[
\bbQ_\perm =
{\rm Span}\Bigg\{
 \sum_{\sigma\in \calS_n} \prod_{i=1}^m \prod_{j=i}^m (r_{\sigma_i \sigma_j})^{\alpha_{ij}}
 :
 m\in\bbN,~
 m\leq n,~
 \balpha\in\bbN^{m\times m}
\Bigg\}
.
\]
\end{lemma}
\begin{proof}
If $q\in\bbQ_\perm$ then
\[
\frac{1}{n!}\sum_{\sigma\in\calS_n} q(r_{\sigma\sigma})
=
\frac{1}{n!}\sum_{\sigma\in\calS_n} q(r)
= q(r).
\]
It remains to apply this identity to all monomials $q(r) = \prod_{i=1}^m \prod_{j=i}^m r_{ij}^{\alpha_{ij}}$ to obtain the stated result.
\end{proof}

The next step towards proving the main result is the following lemma, where we denote $\calN := \{1,\ldots,n\}$:
\begin{lemma}
For $m\in\bbN$, $m\leq n$, and $\balpha\in\bbN^{m\times m}$,
\[
B_\balpha(u)
=
\sum_{\gamma\in\calN^m}
\prod_{i=1}^m
\prod_{j=i}^m
(u_{\gamma_i}\cdot u_{\gamma_j})^{\alpha_{ij}}
.
\]
\end{lemma}
\begin{proof}
Before commencing with the proof, we note that we will use the distributive law of addition and multiplication in the following form:
\[
\prod_{j=1}^m \sum_{i\in\calI} f(i,j)
=
\sum_{\gamma\in\calI^m} \prod_{j=1}^m f(\gamma_j,j).
\]

Then
\begin{align*}
T^{(i)}_{\beta^{(1,i)}\ldots\beta^{(i-1,i)}\beta^{(i,i+1)}\ldots\beta^{(i,m)}}(u)
&=
\sum_{\gamma\in\calN}
|u_\gamma|^{2\alpha_{ii}}
\Bigg(
	\prod_{j=1}^{i-1} \big(u^{\otimes\alpha_{ji}}_{\gamma}\big)_{\beta^{(j,i)}}
\Bigg)
\Bigg(
	\prod_{j=i+1}^{m} \big(u^{\otimes\alpha_{ij}}_{\gamma}\big)_{\beta^{(i,j)}}
\Bigg)
\\ &=
\sum_{\gamma\in\calN}
|u_\gamma|^{2\alpha_{ii}}
\Bigg(
\prod_{j=1}^{i-1} \prod_{\ell=1}^{\alpha_{ji}} u_{\gamma,\beta^{(j,i)}_\ell}
\Bigg)
\Bigg(
\prod_{j=i+1}^{m} \prod_{\ell=1}^{\alpha_{ij}} u_{\gamma,\beta^{(i,j)}_\ell}
\Bigg)
.
\end{align*}

We recall the notations $\alpha'_i$ and $\calB$, introduced in \eqref{eq:alpha-prime} and \eqref{eq:calMcalB} respectively,
and hence express
\begin{align*}
\aprod\balpha_{i=1}^m M_{\alpha_{ii},\alpha_i'}(u)
=&
\sum_{\beta \in \calB}
\prod_{i=1}^m
T^{(i)}_{\beta^{(1,i)}\ldots\beta^{(i-1,i)}\beta^{(i,i+1)}\ldots\beta^{(i,m)}}
\\=&
\sum_{\beta \in \calB}
\prod_{i=1}^m
\sum_{\gamma\in\calN}
|u_\gamma|^{2\alpha_{ii}}
\Bigg(
	\prod_{j=1}^{i-1} \prod_{\ell=1}^{\alpha_{ji}} u_{\gamma,\beta^{(j,i)}_\ell}
\Bigg)
\Bigg(
\prod_{j=i+1}^{m} \prod_{\ell=1}^{\alpha_{ij}} u_{\gamma,\beta^{(i,j)}_\ell}
\Bigg)
\\=&
\sum_{\beta \in \calB}
\sum_{\gamma\in\calN^m}
\prod_{i=1}^m
|u_{\gamma_i}|^{2\alpha_{ii}}
\Bigg(
\prod_{j=1}^{i-1} \prod_{\ell=1}^{\alpha_{ji}} u_{{\gamma_i},\beta^{(j,i)}_\ell}
\Bigg)
\Bigg(
\prod_{j=i+1}^{m} \prod_{\ell=1}^{\alpha_{ij}} u_{{\gamma_i},\beta^{(i,j)}_\ell}
\Bigg)
\\=&
\sum_{\beta \in \calB}
\sum_{\gamma\in\calN^m}
\prod_{i=1}^m
|u_{\gamma_i}|^{2\alpha_{ii}}
\prod_{j=i+1}^{m}
\prod_{\ell=1}^{\alpha_{ij}}
\big(u_{{\gamma_i},\beta^{(i,j)}_\ell}\big)
\big(u_{{\gamma_j},\beta^{(i,j)}_\ell}\big)
\\=&
\sum_{\gamma\in\calN^m}
\prod_{i=1}^m
|u_{\gamma_i}|^{2\alpha_{ii}}
\prod_{j=i+1}^{m}
\sum_{~\beta \in \{1,\ldots,d\}^{\alpha_{ij}}~}
\prod_{\ell=1}^{\alpha_{ij}} 
\big(u_{{\gamma_i},\beta_\ell}\big)
\big(u_{{\gamma_j},\beta_\ell}\big)
\\=&
\sum_{\gamma\in\calN^m}
\prod_{i=1}^m
|u_{\gamma_i}|^{2\alpha_{ii}}
\prod_{j=i+1}^{m}
\prod_{\ell=1}^{\alpha_{ij}}
\sum_{\beta=1}^d
(u_{{\gamma_i},\beta})
(u_{{\gamma_j},\beta})
\\=&
\sum_{\gamma\in\calN^m}
\prod_{i=1}^m
|u_{\gamma_i}|^{2\alpha_{ii}}
\prod_{j=i+1}^{m}
\prod_{\ell=1}^{\alpha_{ij}}
u_{{\gamma_i}} \cdot u_{{\gamma_j}}
\\=&
\sum_{\gamma\in\calN^m}
\prod_{i=1}^m
|u_{\gamma_i}|^{2\alpha_{ii}}
\prod_{j=i+1}^{m}
\big(u_{{\gamma_i}} \cdot u_{{\gamma_j}}\big)^{\alpha_{ij}}
\\=&
\sum_{\gamma\in\calN^m}
\prod_{i=1}^m
\prod_{j=i}^{m}
\big(u_{{\gamma_i}} \cdot u_{{\gamma_j}}\big)^{\alpha_{ij}}
.
\end{align*}
\end{proof}

\begin{proof}[Proof of Theorem \ref{th:main}]
In view of the previous lemma, we can denote
\[
\tilde{\bbQ}_\perm =
{\rm Span}\Bigg\{
\sum_{\gamma\in\calN^m}
\prod_{i=1}^m \prod_{j=i}^m
(r_{\gamma_i \gamma_j})^{\alpha_{ij}}
 :
 m\in\bbN,~
 m\leq n,~
 \balpha\in\bbN^{m\times m}
\Bigg\}
\]
and, applying the earlier lemmas, formulate the statement of Theorem \ref{th:main} as $\tilde{\bbQ}_\perm = \bbQ_\perm$. The latter is an immediate corollary of the following more specialized result.
\end{proof}

\begin{lemma}\label{lem:main_lemma}
For $m\in\bbN$, $m\leq n$, denote
\begin{align*}
	\bbQ^{(m)}_{\perm}
&:= 
	{\rm Span}\Bigg\{
	 \sum_{\sigma\in \calS_n}
	 \prod_{i=1}^m \prod_{j=i}^m
	 (r_{\sigma_i \sigma_j})^{\alpha_{ij}}
	 :
	 \balpha\in\bbN^{m\times m}
	\Bigg\}
,\qquad\text{and}\\
	\tilde{\bbQ}^{(m)}_\perm
&:=
	{\rm Span}\Bigg\{
	\sum_{\gamma\in\calN^m}
	\prod_{i=1}^m \prod_{j=i}^m
	(r_{\gamma_i\gamma_j})^{\alpha_{ij}}
	 :
	 \balpha\in\bbN^{m\times m}
	\Bigg\}
.
\end{align*}
Then $\tilde{\bbQ}^{(m)}_\perm = \bbQ^{(m)}_\perm$.
\end{lemma}

Before we prove this lemma, we give two auxiliary results.
\begin{lemma}
We equip $\calN^m$ with the lexicographical order and hence denote by $\Gamma := \{\gamma\in\calN^m : \gamma = \min\{\sigma_\gamma : \sigma\in\calS_n\}\}$ the set of representatives of equivalence classes $\{\sigma_\gamma : \sigma\in\calS_n\}$, where $\sigma_\gamma$ is understood as composition of tuples \eqref{eq:tuples_composition}. Also let
$C_{\gamma} = \#\{\sigma_\gamma : \sigma\in\calS_n\}$, where $\#$ denotes the number of elements in a set.
Then
\[
	\sum_{\gamma\in\calN^m}
	\prod_{i=1}^m \prod_{j=i}^m
	(r_{\gamma_i\gamma_j})^{\alpha_{ij}}
=
	\sum_{\gamma\in\Gamma}
	\frac{C_\gamma}{n!}
	\sum_{\sigma\in\calS_n}
	\prod_{i=1}^m \prod_{j=i}^m
	(r_{\sigma_{\gamma_i}\sigma_{\gamma_j}})^{\alpha_{ij}}
.
\]
\end{lemma}
\begin{proof}
Any $\sigma\in\calS_n$ induces a one-to-one mapping $\gamma\mapsto\sigma_\gamma$ on $\calN^m$.
Hence 
\[
	\sum_{\gamma\in\calN^m}
	\prod_{i=1}^m \prod_{j=i}^m
	(r_{\gamma_i\gamma_j})^{\alpha_{ij}}
=
	\sum_{\gamma\in\calN^m}
	\prod_{i=1}^m \prod_{j=i}^m
	(r_{\sigma_{\gamma_i}\sigma_{\gamma_j}})^{\alpha_{ij}}
\qquad\forall \sigma\in\calS_n
,
\]
therefore
\[
	\sum_{\gamma\in\calN^m}
	\prod_{i=1}^m \prod_{j=i}^m
	(r_{\gamma_i\gamma_j})^{\alpha_{ij}}
=
	\sum_{\gamma\in\calN^m}
	\frac{1}{n!}
	\sum_{\sigma\in\calS_n}
	\prod_{i=1}^m \prod_{j=i}^m
	(r_{\sigma_{\gamma_i}\sigma_{\gamma_j}})^{\alpha_{ij}}
.
\]
It remains to group up the terms for which $\sigma_\gamma$ is the same.
\end{proof}

The next auxiliary result is proved by a trivial combinatorial argument, essentially expressing that all elements of $\Gamma$ other than $(1,\ldots,m)\in\Gamma$ have repeated values.
\begin{proposition}\label{prop:remaining}
Let $m \geq 1$. Then $\Gamma = \{(1,\ldots,m)\}\cup \Gamma'$, where $\Gamma' := \{\gamma\in\Gamma: \max_i\gamma_i\leq m-1\}$.
\qed
\end{proposition}

\begin{proof}[Proof of Lemma \ref{lem:main_lemma}]
We argue by induction over $m$. For $m=0$ the statement is obvious since $\bbQ^{(0)}_\perm = \tilde{\bbQ}^{(0)}_\perm = {\rm Span}\{1\}$, therefore we only need to prove the induction step.

Now, for $m\in\bbN$, $m\leq n$, we choose an arbitrary $\alpha$ and let
\begin{align*}
q(\br) := 
	 \sum_{\sigma\in S_n}
	 \prod_{j,k=1}^m
	 (r_{\sigma_j \sigma_k})^{\alpha_{jk}}
&\in\bbQ_\perm
,\qquad\text{and} \\
\tilde{q}(\br) := \frac{n!}{C_{\{1,\ldots,m\}}}
	\sum_{\gamma\in\calN^m}
	\prod_{j,k=1}^m
	(r_{\gamma_j\gamma_k})^{\alpha_{jk}}
&\in\tilde{\bbQ}_\perm
.
\end{align*}
These $q(\br)$ and $\tilde{q}(\br)$ span $\bbQ^{(m)}_\perm$ and $\tilde{\bbQ}^{(m)}_\perm$, respectively.
We aim to show that $q(\br) - \tilde{q}(\br) \in \bbQ^{(m-1)}_\perm = \tilde{\bbQ}^{(m-1)}_\perm$.
This will prove the result considering that, by definition, $\bbQ^{(m-1)}_\perm \subseteq \bbQ^{(m)}_\perm$ and $\tilde{\bbQ}^{(m-1)}_\perm \subseteq \tilde{\bbQ}^{(m)}_\perm$.

Indeed, in view of the two previous results, recall the definition $\Gamma' := \{\gamma\in\Gamma: \max_i\gamma_i\leq m-1\}$, and write
\begin{align*}
\tilde{q}(\br) - q(\br)
&=
	\frac{n!}{C_{\{1,\ldots,m\}}}
	\sum_{\gamma\in\Gamma}
	\frac{C_\gamma}{n!}
	\sum_{\sigma\in\calS_n}
	\prod_{i=1}^m \prod_{j=i}^m
	(r_{\sigma_{\gamma_i}\sigma_{\gamma_j}})^{\alpha_{ij}}
-
	 \sum_{\sigma\in S_n}
	\prod_{i=1}^m \prod_{j=i}^m
	 (r_{\sigma_i \sigma_j})^{\alpha_{ij}}
\\ &=
	\sum_{\gamma\in\Gamma'}
	\frac{C_\gamma}{C_{\{1,\ldots,m\}}}
	\sum_{\sigma\in\calS_n}
	\prod_{i=1}^m \prod_{j=i}^m
	(r_{\sigma_{\gamma_i}\sigma_{\gamma_j}})^{\alpha_{ij}}
.
\end{align*}
Next, we denote
\[
\alpha^{(\gamma)}_{k\ell}:= \sum_{\substack{1\leq i\leq j\leq m \\ \gamma_i = k,~ \gamma_j = \ell}} \alpha_{ij}
\]
and hence express
\begin{align*}
\tilde{q}(\br) - q(\br)
&=
	\sum_{\gamma\in\Gamma'}
	\frac{C_\gamma}{C_{\{1,\ldots,m\}}}
	\sum_{\sigma\in\calS_n}
	\prod_{k=1}^{m-1} \prod_{\ell=i}^{m-1}
	(r_{\sigma_{k}\sigma_{\ell}})^{\alpha^{(\gamma)}_{k\ell}}
.
\end{align*}
Note that the upper limit of both products is $m-1$ thanks to the definition of $\Gamma'$ (and, of course, Proposition \ref{prop:remaining}).
Also, note that we used the fact that, from the definition of $\alpha^{(\gamma)}$, $\alpha^{(\gamma)}_{k\ell}=0$ whenever $k>\ell$.
Since $\sum_{\sigma\in\calS_n}
	\prod_{k=1}^{m-1} \prod_{\ell=i}^{m-1}
	(r_{\sigma_{k}\sigma_{\ell}})^{\alpha^{(\gamma)}_{k\ell}}$
by definition belongs to $\bbQ^{(m-1)}_\perm$ for each $\gamma$, we finally derive that
\[
\tilde{q}(\br) - q(\br)\in
\bbQ^{(m-1)}_\perm
=
\tilde{\bbQ}^{(m-1)}_\perm
.
\]
This concludes the proof of the induction step.
\end{proof}

\section{Practical Implementation}\label{sec:implementation}

%
%
%

The representation of the interatomic potential through polynomials, outlined in Section \ref{sec:representation}, does not satisfy the {\bf (R3)} property needed for the practical implementation.
Hence, as the next step we modify the interatomic potential to satisfy this property.
After this, in Section \ref{sec:implementation:algorithm} we discuss the steps needed to compute $B_\alpha(u)$, and in Section \ref{sec:implementation:training} we describe the algorithms we used for fitting the potentials.

First, notice that for a fixed $\nu$, a linear combination of moment tensors $M_{\mu,\nu}(u)$ is a polynomial of $|u_i|^2$ multiplied by $u_i^{\otimes\nu}$.
The space of polynomials of $|u_i|^2$ can be substituted with any other space of functions (which, for generality, can be made dependent on $\nu$), provided that they can represent any regular function of $|u|$, i.e.,
\begin{equation}\label{eq:Mtilde}
\tilde{M}_{\mu,\nu}(u) := \sum_{i=1}^n f_{\mu,\nu}(|u_i|) u_i^{\otimes \nu},
\end{equation}
where, e.g., $f_{\mu,\nu}(r) = r^{-\mu-\nu} f_\cut(r)$ or $f_\mu(r) = e^{-k_\mu r} f_\cut(r)$, $k_\mu>0$ is some sequence of real numbers, and $f_\cut(r)$ is some cutoff function such that $f_\cut(r)=0$ for $r\geq R_\cut$.
Here $f_{\mu,\nu}(r)$ plays essentially the same role as the ``radial symmetry functions'' in the NNP \cite{Behler2014review} or the radial basis functions in \cite{BartokKondorCsanyi2013descriptors}.
We then let
\[
\tilde{B}_{\alpha} (u)
:=
\aprod\balpha_{i=1}^k \tilde{M}_{\alpha_{ii},\alpha_i'}(u)
\]
(cf.\ \eqref{eq:Balpha}) and define
the interatomic potential by
\[
V(u) = \sum_{\alpha\in A} c_{\alpha} \tilde{B}_{\alpha}(u),
\]
where $A$ is a set of matrix-valued indices fixed {\it a priori} and $c_{\alpha}$ is the set of coefficients found in the training stage.
For the rest of the paper we will omit tildes in $\tilde{M}_{\bullet,\bullet}$ and $\tilde{B}_\bullet$.


\subsection{Computing the Energy and Forces}\label{sec:implementation:algorithm}

Next, we discuss the steps needed to compute the interatomic potential and its derivatives for a given neighborhood $u$.
The computation consists of two parts, the precomputation (offline) step and the evaluation (online) step.
The precomputation step accepts the set $A$ of values of $\alpha$ as an input and generates the data for the next step, which is the efficient calculation of $B_\alpha(u)$ for a given neighborhood $u$.

Before we proceed, we make two observations.
\begin{enumerate}
\item 

The elements of the tensors $\tildeM_{\mu,\nu}$ are the ``moments'' of the form
\[
m_{\mu,\beta}(u) := \sum_{i=1}^n f_{\mu,\nu}(|u_i|) \prod_{j=1}^d u_{i,j}^{\beta_j},
\]
where $\beta$ is a multiindex such that $|\beta|=\nu$.
The higher the dimension is, the more repeated moments each $\tildeM_{\mu,\nu}$ contains (e.g., each matrix $\tildeM_{\mu,2}$ has 9 elements, out of which at most 6 may be different due to the symmetricity of $\tildeM_{\mu,2}$).

\item 

The scalar functions $B_{\alpha}$ consist of products of $\tildeM_{\mu, \nu}$ (which means that the elements of $B_{\alpha}$ are linear combinations of products of $m_{\mu,\beta}$).
Differentiating products of two terms is easier than products of three or more terms.
Hence it will be helpful to have a representation of $B_\alpha$ as a product of two tensors.

To that end, we extend the definition of the product $\aprod{\alpha}$ by allowing the result to be a tensor of nonzero dimension:
\[
\bigg(
\aprod{\alpha}_{i=1}^{k} T^{(i)}
\bigg)_{\beta^{(1,1)}\ldots\beta^{(k,k)}}
:=
\sum_{\beta\in\calB}
\prod_{i=1}^k T^{(i)}_{\beta^{(i,1)}\ldots\beta^{(i,k)}},
\]
where each $T^{(i)}$ is a tensor of dimension $\sum_{j=1}^{k} \alpha_{ij}$ (here we let $\alpha_{ij}=\alpha_{ji}$), $\calB$ is defined in \eqref{eq:calMcalB}, and we make a convention that $\beta^{(i,j)}:=\beta^{(j,i)}$ for $i>j$.

Hence we define
\[ 
\hat{B}_{\bar{\alpha},\alpha}(u)
:=
\aprod{\alpha}_{i=1}^k \tildeM_{\bar{\alpha}_i, \sum_{j=1}^k \alpha_{ij}}(u),
\]
parametrized by $k\in\bbN$, $\bar{\alpha}\in\bbN^k$ and a symmetric matrix  $\alpha\in\bbN^{k\times k}$.
$\hat{B}_{\bar{\alpha},\alpha}(u)$ is a tensor of dimension $\sum_{i=1}^k \alpha_{ii}$.
Clearly if $\bar{\alpha}_i = \alpha_{ii}$ for all $i$ then $\hat{B}_{\bar{\alpha},\alpha} = B_{\alpha}$, which makes the collection of tensors $\hat{B}_{\bar{\alpha},\alpha}$ a generalization of $B_\alpha$.

Next, consider $\hat{B}_{\bar{\alpha},\alpha}(u)$ for some $\bar{\alpha}\in\bbN^k$ and $\alpha\in\bbN^{k\times k}$, fix $1\leq I \leq k$, and denote
$\bar{\beta} = (\bar{\alpha}_1,\ldots,\bar{\alpha}_I)$,
$\bar{\gamma} = (\bar{\alpha}_{I+1},\ldots,\bar{\alpha}_k)$,

\begin{align*}
\beta
&:=
\begin{pmatrix}
\alpha_{11} + \sum_{i=I+1}^k \alpha_{i1} & \alpha_{12} & \cdots & \alpha_{1I} \\
\alpha_{12} & \alpha_{22} + \sum_{i=I+1}^k \alpha_{i2} & \cdots & \alpha_{2I} \\
\vdots & \vdots & \ddots & \vdots \\
\alpha_{1I} & \alpha_{2I} & \cdots & \alpha_{II} + \sum_{i=I+1}^k \alpha_{iI}
\end{pmatrix},
\qquad\text{and}
\\
\gamma
&:=
\begin{pmatrix}
\alpha_{I+1,I+1} + \sum_{i=1}^I \alpha_{i,I+1} & \alpha_{I+1,I+2} & \cdots & \alpha_{I+1,k} \\
\alpha_{I+1,I+2} & \alpha_{I+2,I+2} + \sum_{i=1}^I \alpha_{i,I+2} & \cdots & \alpha_{I+2,k} \\
\vdots & \vdots & \ddots & \vdots \\
\alpha_{I+1,k} & \alpha_{I+2,k} & \cdots & \alpha_{kk} + \sum_{i=1}^I \alpha_{i,k}
\end{pmatrix}.
\end{align*}

One can then express the elements of $\hat{B}_{\bar{\alpha},\alpha}(u)$ through products of elements of $\hat{B}_{\bar{\beta},\beta}(u)$ and $\hat{B}_{\bar{\gamma},\gamma}(u)$.
Note that by reordering the rows and columns of $\alpha$ and $\beta$, one can generate many different ways of representing $\hat{B}_{\bar{\alpha},\alpha}(u)$.
When doing computations, one should exercise this freedom in such a way that the resulting tensors are of minimal dimension so that the total computation cost is reduced.

Finally, note that even if $\hat{B}_{\bar{\alpha},\alpha}(u)$ was scalar-valued, $\hat{B}_{\bar{\beta},\beta}(u)$ and $\hat{B}_{\bar{\gamma},\gamma}(u)$ do not have to be scalar-valued---this motivates the need to introduce the tensor-valued basis functions $\hat{B}_{\bar{\alpha},\alpha}(u)$.
\end{enumerate}

\subsubsection*{Precomputation}
The precomputation step is hence as follows (we keep the argument $u$ of, e.g., in $B_\alpha(u)$ to match the above notation, however, the particular values of $u$ never enter the precomputation step)
\begin{itemize}
\item
	Starting with the set of tensors $B_{\alpha}(u)$ (indexed by $\alpha\in A$), establish their correspondence to $\hat{B}_{\bar{\alpha},\alpha}(u)$ and recursively represent each $\hat{B}_{\bar{\alpha},\alpha}(u)$ through some $\hat{B}_{\bar{\beta},\beta}(u)$ and $\hat{B}_{\bar{\gamma},\gamma}(u)$ as described above.
	In each case, out of all such products choose the one with the minimal sum of the number of dimensions of these two tensors.
\item
	Enumerate all the elements of all the tensors $\hat{B}_{\bar{\alpha},\alpha}(u)$ as $b_i(u)$ ($i$ is the ordinal number of the corresponding element).
	
\item
	Represent each $b_i(u)$ as either $b_i(u) = m_{\mu_i, \beta_i}(u)$ or $b_i(u) = \sum_{j=1}^{J_i} c_j b_{\ell_j}(u) b_{k_j}(u)$.

\item
	Output the resulting
	\begin{itemize}
	\item[(1)] correspondence of $B_\alpha(u)$ to $b_i(u)$,
	\item[(2)] $\mu_i$ and $\beta_i$, and
	\item[(3)] tuples of $(i,c,\ell,k)$ corresponding to $b_i(u) = \sum_{j=1}^{J_i} c_j b_{\ell_j}(u) b_{k_j}(u)$.
	\end{itemize}
\end{itemize}

\subsubsection*{Evaluation}

The evaluation step, as written out below, accepts $u$ as an input and evaluates $B_\alpha(u)$ using the precomputed data (described above).
\begin{enumerate}
\item For a given $u$, calculate all $m_{\mu_i,\beta_i}(u)$.
\item Then calculate all other $b_i(u)$ using the tuples of $(i,c,\ell,k)$.
\item Finally, pick those $b_i(u)$ that correspond to scalar $B_\alpha(u)$ (as opposed to non-scalar $\hat{B}_{\bar{\alpha}, \alpha}$).
\end{enumerate}

It remains to form the linear combination of $B_\alpha(u)$ with the coefficients obtained from a linear regression (training), and then sum up all these linear combinations for all atomic environments to obtain the interatomic interaction energy of a given atomistic system.
The forces are computed by reverse-mode differentiation of the energy with respect to the atomic positions \cite{backprop}.

\subsection{Training}\label{sec:implementation:training}

Once the set $A$ of values of $\alpha$ is fixed, we need to determine the coefficients $c_{\alpha}$.
This is done with the regularized (to avoid overfitting \cite{Alpaydin2014MachineLearning}) linear regression in the following way.

Let a database of atomic configurations $X = \{x^{(k)}:k=1,\ldots,K\}$, where $x^{(k)}$ is of size $N^{(k)}$, be given together with their reference energies and forces, $E^\q(x^{(k)}) = E^{(k)}$ and $-\nabla E^\q(x^{(k)}) = f^{(k)}$.
We form an overdetermined system of linear equations on $c_{\alpha}$,
\begin{align*}
\sum_{i=1}^{N^{(k)}} \sum_{\alpha\in A} c_\alpha B(Dx^{(k)}_i) &= E^{(k)}
,
\\
\frac{\partial}{\partial x^{(k)}_j}\sum_{i=1}^{N^{(k)}} \sum_{\alpha\in A} c_\alpha B(Dx^{(k)}_i) &= -f^{(k)}_j,
\end{align*}
(cf.\ \eqref{eq:D_def} and \eqref{eq:E_def}), which we write in the matrix form $X c = g$.
These equations may be ill-conditioned, hence a regularization must be used.
We tried three versions of regularization, namely the $\ell_p$ regularization with $p\in\{0,1,2\}$, all described below.
All three can be written as
\[
\text{find~~~~} \min_c \|Xc-g\|^2
\text{~~~~subject to~~~~}
\|c\|_{\ell_p}^2 \leq t,
\]
where $t$ is the regularization parameter and $\|c\|_{\ell_0}$ is defined as the number of nonzero entires in $c$.
For $p\geq 1$ this can be equivalently rewritten as
\[
\text{find~~~~} \min_c \|Xc-g\|^2 + \gamma \|c\|_{\ell_p}^2,
\]
where $\gamma$ is an alternative regularization parameter.

\subsubsection{$\ell_2$ regularization}

For the solution of the overdetermined linear equations $X c = g$ we take
\[
c = (X^T X + \gamma\, {\rm diag}(X^T X))^{-1} X^T g,
\]
where ${\rm diag}(A)$ denotes the diagonal matrix whose diagonal elements are the same as in $A$.
The penalization matrix was chosen as ${\rm diag}(X^T X)$ instead of the identity matrix so that its scaling with respect to the database size and the scale of the basis functions $B_\alpha$, is compatible with that of the covariance matrix $X^T X$.
The regularization parameter $\gamma$ was determined from the 16-fold cross-validation scheme \cite{Alpaydin2014MachineLearning} as described in the next paragraph.

To perform the 16-fold cross-validation, we split the database $X$ evenly into 16 non-overlapping databases $\tilde{X}_1,\ldots,\tilde{X}_{16}$.
We then train 16 different models, each on the database $X\setminus \tilde{X}_i$ and find the RMS error when tested on $\tilde{X}_i$.
The parameter $\gamma$ is then chosen such that the cross-validation RMS error averaged over these 16 models is minimal.

\subsubsection{$\ell_0$ regularization}\label{sec:l0}

The advantage of the $\ell_0$ regularization is that it produces sparse solutions $c$, whereas the $\ell_2$ regularization does not.
We note that the $\ell_1$ regularization also produces sparse solutions \cite{Tibshirani1996lasso}, but our numerical experiments show that the $\ell_0$ regularization produces significantly more sparse solutions (however, at a cost of a larger precomputation time).

We thus solve a sequence of problems, parametrized by an integer parameter $N_{\rm nz}$ (number of non-zeros) as follows:
\[
\text{find }\min_c \|Xc-g\|
\text{ subject to } \|c\|_{\ell_0} = N_{\rm nz},
\]
and choose the minimal $N_{\rm nz}$ such that $\|Xc-g\|$ reaches the accuracy goal.

To describe the algorithm, it is convenient to rewrite the problem as follows.
Let $\bm{A}$ be the set of all indices $\alpha$ (earlier denoted as $A$).
\begin{align*}
\text{find }A\subset \bm{A}
& \text{ such that } \min_c \|Xc-g\| \text{ is minimal,} \\
& \text{ subject to } |A| = N_{\rm nz}
\text{ and } c_{\bm{A}\setminus A} = 0.
\end{align*}

This is essentially a compressed sensing problem \cite{Zhang2013_compressed_sensing}.
In order to solve it we take a standard greedy algorithm (similar to the matching pursuit from the compressed sensing literature) and turn it into a genetic algorithm by adding the local search and crossover steps.
The main variable in this algorithm is the family (population) of the sets $A$ which is denoted by $\calA$.
The cap on the population size is set to be $\Ncap>1$.
The algorithm is as follows.

\begin{itemize}
	\item[1.] Let $\calA=\{\emptyset\}$ as the solution for $N_{\rm nz}=0$.
	
	\item[2.] For each $A\in\calA$ find $i\notin A$, such that $c$ corresponding to $\calA\cup \{i\}$ is the best (i.e., minimizing $\|Xc-g\|$). Then replace $A$ with $A\cup \{i\}$.
	
	\item[3.] ``Crossover'':
		If $|\calA|>1$ then do the following.
		For each pair of sets, $A,A'\in\calA$, divide randomly these sets into two nonintersecting subsets, $A = A_1\cup A_2$ and $A' = A_1'\cup A_2'$, and generate new sets $A_1\cup A_2'$ and $A_1'\cup A_2$.
		To generate such splittings of the sets, first sample uniformly an integer $m\in\{1,\ldots, |A\setminus A'|\}$ and then form $A_2$ and $A_2'$ by uniformly sampling $m$ distinct elements from $A$ and $A'$ respectively.
		Then replace all the sets in $\calA$ with the newly generated sets. (Note that if $|\calA|=\Ncap$ then up to $\Ncap(\Ncap-1)$ sets will be generated---2 sets from each of $\binom{\Ncap}{2}$ pairs.)
	
	\item[4.] ``Local search'':
	\begin{itemize}
		\item[4.1.] For each $A\in\calA$ find $j\in A$ such that $c$ corresponding to $A\setminus \{j\}$ is the best.
		\item[4.2.] Then find $i\notin A\setminus \{j\}$ such that $c$ corresponding to $(A\setminus \{j\})\cup\{i\}$ is the best.
		\item[4.3.] If $i\ne j$ then:
		\begin{itemize}
			\item[4.3a.] include $(A\setminus \{j\})\cup\{i\}$ into $\calA$,
			\item[4.3b.] if $|\calA|> \Ncap$ then exclude $A$ from $\calA$, and
			\item[4.3c.] go to step 4.1.
		\end{itemize}
	\end{itemize}
	
	\item[5.] Remove all but $\Ncap$ best sets in $\calA$.
	
	\item[6.] Repeat steps 2--5 until the accuracy goal on $\|Xc-g\|$ is reached.
	Then take the best set $A\in\calA$ and compute the corresponding $c$.
\end{itemize}

We note that whenever $A$ is fixed then finding $c$ corresponding to $A$ is easy:\\
$c = ((X^T X)_{AA})^{-1} (X^T g)_A,$
where $\bullet_A$ and $\bullet_{AA}$
are the operations of extracting a subset of rows and columns corresponding to $A\subset \bm{A}$.

\section{Numerical Experiments}\label{sec:performance-tests}

Our next goal is to understand how MTP performs compared to other interatomic potentials.
Unfortunately, there are no existing works performing quantitative comparison between different machine learning potentials in terms of their accuracy and computational efficiency.
In the present work we compare the performance of MTP with that of GAP for tungsten \cite{Szlachta2014thesis,2014GAP_tungsten} on the QM database published at {\tt www.libatoms.org} together with the GAP code. 
We test MTP by fitting it on the database of 9\,693 configurations of tungsten, with nearly 150\,000 individual atomic environments, and compare it to the analogously fitted GAP, namely, the one tagged as GAP$_6$ in \cite{2014GAP_tungsten} or as the iterative-SOAP-GAP or I-S-GAP in \cite{Szlachta2014thesis}.
We note that this is a database of the Kohn-Sham DFT calculations (as opposed to a tight-binding model used in the analysis) for an electronic temperature of 1000$^\circ$K. We did not prove algebraic convergence of MTP to this model, however we will observe it numerically.

We choose $R_\cut=4.9$\AA, and also set the minimal distance parameter to be $R_{\rm min} := 1.9$\AA.
For the radial functions we choose
\[
\hat{f}_{\mu,\nu}(r) := \begin{cases}
r^{-\nu-2} r^\mu (R_\cut-r)^2 & r<R_\cut, \\
0 & r\geq R_\cut
,
\end{cases}
\]
and then for each $\nu$ we orthonormalize them on the interval $[R_{\rm min},R_\cut]$ with the weight $r^{2\nu}(r-R_{\rm min})(R_\cut-r)$. This procedure yields us the functions $f_{\mu,\nu}$ used with \eqref{eq:Mtilde}.
(This procedure is equivalent to first orthonormalizing the functions $r^{-2} r^\mu (R_\cut-r)^2$ with the weight $(r-R_{\rm min})(R_\cut-r)$, the same weight as for the Chebyshev polynomials, and then multiplying by $r^{-\nu}$.)
Here $r^{-\nu}$ compensates for $r^{\otimes\nu}$, $r^{-2}$ prioritizes closer atoms to more distant onces, and $(R_\cut-r)^2$ ensures a smooth cut-off.

\subsection{Convergence with the Number of Basis Functions}

Even though theoretically it is proved that MTP is systematically improvable, the actual accuracy depends on the size of the regression problem to be solved.
We therefore first study how fast the MTP fit converges with the number of basis functions used.

Theorem \ref{th:conv} suggests that the fitting error decreases exponentially with the polynomial degree. At the same time, a crude upper bound on the dimension of the space of polynomials of degree $m$ is $(n+1)^m$, where $n$ is the maximal number of atoms in an atomic neighborhood.
This suggests an algebraic decay of the fitting error with the number of basis functions.

We hence study convergence of the error of fitting of potentials based on two choices of the sets of $\alpha$.
The first choice of the set of $\alpha$, $A_N$, is to limit the corresponding degree of $B_\alpha$ for $\alpha\in A_N$ by $N$:
\[
A_N = \{\alpha:\deg(B_\alpha) \leq N\}.
\]
Since $f_{\mu,\nu}$ are no longer polynomials, we make a convention that the degree of $f_{\mu,\nu}$ is $\mu+1$ (while the degree of the constant function $f(u) := 1$ is zero).
Interestingly, it was found that a slightly better choice is
\[
A'_N = \{\alpha:\deg(B_\alpha) \leq N+8(\#\alpha)\},
\]
where $\#\alpha$ is the length of $\alpha$.
In fact, $A'_N$ corresponds to $A_N$ if we make another convention that $\deg(f_{\mu,\nu}) = \mu+5$.

Figure \ref{fig:error} displays the RMS fitting error in forces as a function of the size of the set $A_N$.
One can observe an algebraic convergence, which indicates that Theorem \ref{th:conv} is valid for the Kohn-Sham DFT also.
The observed rate is $(\# A)^{-0.227}$, where the exponent $-0.227$ is not a universal constant associated with MTP, but depends on the database chosen.
A preasymptotic regime with a faster algebraic convergence rate can be seen for small $\#A$.

\begin{figure}[htb]
\hfill
\includegraphics[scale=1]{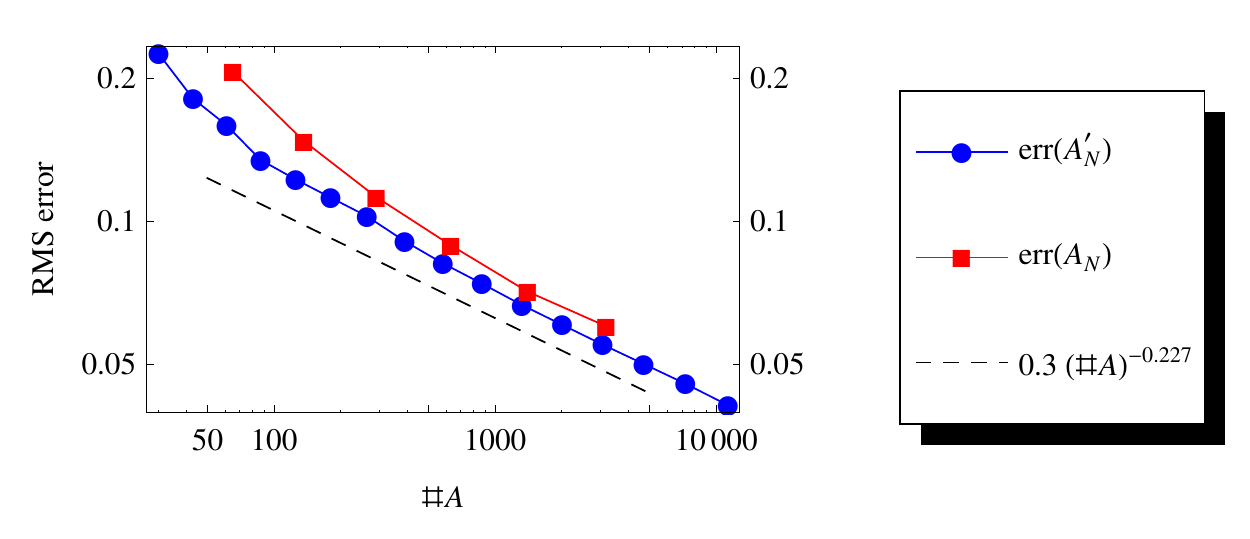}
\hfill
$\mathstrut$
\caption{%
	RMS fitting error in forces as a function of the size of $A_N$ and $A'_N$.
	An algebraic convergence can be observed.
}
\label{fig:error}
\end{figure}


\subsection{Performance Tests}

We next test the performance of MTP and compare it to GAP in terms of fitting error and computation (CPU) time.
We choose two versions of MTP different from each other by the set $A$.

The first MTP, denoted by MTP$_1$, is generated by limiting $\deg(B_\alpha) + 8(\#\alpha) \leq 62$ and additionally limiting $(\#\alpha)\leq 4$, $\mu\leq 5$, and $\nu \leq 4$.
MTP$_2$ is generated by limiting $\deg(B_\alpha) + 8(\#\alpha) \leq 52$, additionally limiting $(\#\alpha)\leq 5$, $\mu\leq 3$, and $\nu \leq 5$, and applying the $\ell_0$ regularization algorithm with $\Ncap=4$, as described in Section \ref{sec:l0}, to extract $760$ basis functions (so as to make the accuracy of GAP and MTP$_2$ same, as discussed in the next paragraph).


\begin{table}[htb]
\begin{center}
\begin{tabular}{|rccc|} \hline
Potential:          & GAP & MTP$_1$ & MTP$_2$ \\
\hline
CPU time/atom [ms]: & 134.2 $\pm$2.6  & 2.9 $\pm$0.5  & 0.8 $\pm$0.2 \\
basis functions: & 10\,000 & 11\,133 & 760 \\
\hline
\multicolumn{4}{|c|}{Fit errors:}\\ \hdashline
force RMS error [eV/\AA]:  & 0.0633 & 0.0427 & 0.0633 \\
$\mathstrut$[\%]:          & 4.2\%  & 2.8\%  & 4.2\% \\
\hline
 \multicolumn{4}{|c|}{Cross-validation errors:}\\ \hdashline
regularization parameter $\gamma$: & - & $3\cdot10^{-9}$  & 0 \\
force RMS error[eV/\AA]: & - & 0.0511  & 0.0642 \\
$\mathstrut$[\%]:        & - & 3.4\% & 4.3\% \\
\hline
\end{tabular}
\end{center}
\caption{Efficiency and accuracy of MTP as compared to GAP.
The mean CPU time is given for a single-core computation on the laptop Intel i7-2675QM CPU, together with the standard deviation as computed on a number of independent runs.
The last section of the table reports the force RMS errors for the 16-fold cross-validation.
}
\label{tb:compare}
\end{table}

The data from the conducted efficiency and accuracy tests are summarized in Table \ref{tb:compare}.
The RMS force (more precisely, RMS of all force components over all non-single-atom configurations) is 1.505 eV/\AA, which is used to compute the relative RMS error.
The errors relative to this RMS force are also presented in the table.
The GAP error is calculated based on the data from \cite{Szlachta2014thesis}.
The CPU times do not include the initialization (precomputation) or constructing the atomistic neighborhoods.

One can see that MTP$_1$ has about the same number of fitting parameters as GAP, while its fitting accuracy is about 1.5 times better and the computation time is 40 times smaller.
MTP$_2$ was constructed such that its fitting accuracy is the same as GAP, but it uses much less parameters for fitting and its computation is more than two orders of magnitude faster.

Also included in the table is the 16-fold cross-validation error.
It shows that MTP$_2$ is not overfitted on the given database, whereas MTP$_1$ needs regularization to avoid overfitting.
We anticipate, however, that by significantly increasing the database size, the cross-validation error of MTP$_1$ would go down and reach the current fitting error of MTP$_1$, since the fitting error follows closely the algebraic decay (see Figure \ref{fig:error}) and is not expected to deteriorate with increasing the database size.

%
%


\section{Conclusion}

The present paper considers the problem of approximating a QM interaction model with interatomic potentials from a mathematical point of view.
In particular, (1) a new class of nonparametric (i.e., systematically improvable) potentials satisfying all the required symmetries has been proposed and advantages in terms of accuracy and performance over the existing schemes have been discussed and (2) an algebraic convergence of fitting with these potentials in a simple setting has been proved and it was then confirmed with the numerical experiments.

This work is done under the assumption that all atoms are chemically equivalent.
A straightforward extension to multicomponent systems would be to let the radial functions depend not only on the positions of atoms $x_i$, but also on the types of atoms $t_i$.
Thus, the expression for the moments for atom $i$ could be
\[
\tildeM_{\mu,\nu} =
\sum_{j} f_{\mu,\nu}(|x_j-x_i|, t_i, t_j) (x_j-x_i)^{\otimes \nu},
\]
where the summation is over the neighborhood of atom $i$.
We leave exploring this path to future publications.

\section*{Acknowledgement}

The author is grateful to G\'abor Cs\'anyi for igniting the interest in this topic and for valuable discussions as a part of our on-going collaboration.
Also, the author thanks Albert Bart\'ok-P\'artay for his help with implementation of MTP in the QUIP software\footnote{Available at {\tt http://www.libatoms.org/Home/LibAtomsQUIP}} which was used to compare the performance of MTP with GAP.
Finally, the author thanks the anonymous referees for many suggestions that lead to improvement of the paper.

\bibliographystyle{plain}
\bibliography{report}

\end{document}